\documentclass{amsart}

\usepackage{amssymb,latexsym,amsfonts,amsmath}
\usepackage{graphicx,psfrag,epsfig,epsf}
\usepackage{eso-pic}
\usepackage{graphicx}
\usepackage{color}
\usepackage{diagrams}
\usepackage{stackrel}
\usepackage{algorithm2e}
\usepackage{tikz}
\usetikzlibrary{automata}
\usepackage{subfigure}
\usepackage{amsfonts}

\usepackage{amssymb,latexsym,amsfonts,amsmath}
\usepackage{graphicx}

\newtheorem{theorem}{Theorem}[section]
\newtheorem{lemma}[theorem]{Lemma}

\newtheorem{proposition}[theorem]{Proposition}

\theoremstyle{definition}
\newtheorem{definition}[theorem]{Definition}
\newtheorem{example}[theorem]{Example}

\theoremstyle{remark}

\numberwithin{equation}{section}

\newcommand{\Diam}{\mathrm{Diam}}
\newcommand{\Split}{\mathrm{Split}}

\newcommand{\Gran}{\mathrm{Gran}}
\newcommand{\s}{\mathrm{s}}

\newcommand{\Post}{\mathrm{Post}}

\newcommand{\Bisim}{\mathrm{Bisim}}

\begin{document}

\begin{abstract}                    
Symbolic models have been recently used as a sound mathematical formalism for the formal verification and control design of purely continuous and hybrid systems. 
In this note we propose a sequence of symbolic models that approximates a discrete--time Piecewise Affine (PWA) system in the sense of approximate simulation and converges to the PWA system in the so--called simulation metric. Symbolic control design is then addressed with specifications expressed in terms of non--deterministic finite automata. A sequence of symbolic control strategies is derived which converges, in the sense of simulation metric, to the maximal controller solving the given specification on the PWA system. 
\end{abstract}

\title[Symbolic Models and Control of Discrete--Time Piecewise Affine Systems]{Symbolic Models and Control of Discrete--Time Piecewise Affine Systems: An Approximate Simulation Approach}
\thanks{The research leading to these results has been partially supported by the Center of Excellence DEWS and received funding from the European Union Seventh Framework Programme [FP7/2007-2013] under grant agreement n.257462 HYCON2 Network of excellence.}

\author[Giordano Pola]{
Giordano Pola$^{1}$}
\address{$^{1}$
Department of Electrical and Information Engineering, Center of Excellence DEWS,
University of L{'}Aquila, Via G. Gronchi, 67100 L{'}Aquila, Italy}
\email{ \{giordano.pola\}@univaq.it}

\author[Maria D. Di Benedetto]{
Maria D. Di Benedetto$^{2}$}
\address{$^{2}$
Department of Electrical and Information Engineering, Center of Excellence DEWS,
University of L{'}Aquila, Via G. Gronchi, 67100 L{'}Aquila, Italy}
\email{ \{mariadomenica.dibenedetto\}@univaq.it}

\maketitle

\section{Introduction}

Piecewise Affine (PWA) systems have been extensively studied in the past and important research advances have been achieved, which comprise research topics on 
stability and stabilizability, observability, controllability, identification, optimal control and reachability. In spite of a well established literature on PWA systems, it is known that reachability problems for PWA systems are undecidable \cite{WhatsDecidable}. This poses serious difficulties for the formal verification and control design of such systems and spurred some researchers to approach the analysis and control of PWA systems through approximating techniques and in particular, by resorting to symbolic models. A symbolic model of a continuous or hybrid system is a finite state automaton in which a symbolic state corresponds to an aggregate of continuous states and a symbolic control label to an aggregate of continuous control inputs. Symbolic models have been employed in \cite{Lee2012,BeltaTACTN2010,Belta2012} as an effective tool to address stabilizability problems, formal verification and control design of discrete--time PWA systems. 
The work in \cite{Lee2012} explores the use of symbolic models for stabilizability problems while the work in \cite{BeltaTACTN2010} for solving formal verification problems; these papers consider PWA systems with no control inputs. The work in \cite{Belta2012} instead, considers PWA systems with control inputs and uses symbolic models for solving control problems with temporal logic--types specifications. In \cite{Lee2012,BeltaTACTN2010} a sequence of abstractions is proposed which approximates the PWA system in the sense of simulation relations. While being provably correct, the results in \cite{BeltaTACTN2010,Belta2012} do not quantify the conservativeness of the approach in the formal verification and control design of PWA systems. Quantifying conservativeness is important to evaluate how far the solutions based on symbolic models are from the corresponding solutions in the pure hybrid domain. In this note we propose a framework based on the notion of approximate simulation \cite{AB-TAC07}, a generalization of the notion of simulation to metric systems, where the accuracy of the approximation scheme is formally quantified and convergence properties are derived. We define a sequence of symbolic models that approximate a PWA system in the sense of approximate simulation, so that the distance between the symbolic models and the PWA system can be quantified through the notion of simulation metric. 
The sequence is proven to converge in the simulation metric to the PWA system. 
Symbolic control design is then addressed where specifications are expressed in terms of non--deterministic finite automata. We propose a sequence of symbolic control strategies that solve the control design problem with increasing accuracy.  
The sequence is proven to converge in the simulation metric to the maximal controller solving the given specification on the original PWA system. 
An illustrative example is included which shows the main results of the note. 
The present paper presents a mature version of the results appeared in \cite{PolaADHS2012}, which includes proofs and an illustrative example. 

\section{Notation and Preliminary Definitions}\label{sec2}
We denote by $2^{X}$ the set of subsets of a set $X$. 
We identify a binary relation $\mathcal{R}\subseteq X\times Y$ with the map \mbox{$\mathcal{R}:X\rightarrow 2^{Y}$} defined by $y\in \mathcal{R}(x)$ if and only if \mbox{$(x,y)\in \mathcal{R}$}. Given a relation $\mathcal{R}\subseteq X\times Y$, the symbol $\mathcal{R}^{-1}$ denotes the inverse relation of $\mathcal{R}$, i.e. \mbox{$\mathcal{R}^{-1}:=\{(y,x)\in Y \times X:(x,y)\in \mathcal{R}\}$}. A graph is an ordered pair $\mathcal{G}=(\mathcal{N},\mathcal{E})$ comprising a set $\mathcal{N}$ of nodes together with a set $\mathcal{E}\subseteq \mathcal{N}\times \mathcal{N}$ of edges. Graph $\mathcal{G}=(\mathcal{N},\mathcal{E})$ is a subgraph of graph $\mathcal{G}'=(\mathcal{N}',\mathcal{E}')$ if $\mathcal{N}\subseteq \mathcal{N}'$ and $\mathcal{E}\subseteq \mathcal{E}'$. 
A connected component of a graph is a subgraph in which any two nodes are connected to each other by paths, and which is connected to no additional nodes in the original graph. 
The symbols $\mathbb{Z}$, $\mathbb{N}_{0}$, $\mathbb{R}$, $\mathbb{R}^{+}$ and $\mathbb{R}_{0}^{+}$ denote the set of integers, non--negative integers, reals, positive and non--negative reals, respectively. The symbol $\Vert \cdot \Vert$ denotes the infinity norm. Given $i_{1},i_{2}\in\mathbb{N}_{0}\cup\{\infty\}$ with $i_{1}<i_{2}$ we denote by $[i_{1};i_{2}]$ the set $\{i_{1},i_{1}+1,...,i_{2}\}$. 
A polyhedron $P\subseteq \mathbb{R}^{n}$ is a set obtained by the intersection of a finite number of (open or closed) half--spaces. A polytope is a bounded polyhedron. 
Given a set $X$, a function $\mathbf{d}:X\times X\rightarrow \mathbb{R}^{+}_{0}\cup\{\infty\}$ is a quasi--pseudo--metric for $X$ if 
(i) for any $x\in X$, $\mathbf{d}(x,x)=0$ and (ii) for any $x,y,z\in X$, $\mathbf{d}(x,y)\leq \mathbf{d}(x,z)+\mathbf{d}(z,y)$. 
If condition (i) is replaced by (i') $\mathbf{d}(x,y)=0$ if and only if $x=y$, then $\mathbf{d}$ is said to be a quasi--metric for $X$. If function $\mathbf{d}$ enjoys properties (i), (ii) and property (iii) for any $x,y\in X$, $\mathbf{d}(x,y)=\mathbf{d}(y,x)$,
then $\mathbf{d}$ is said a pseudo--metric for $X$. If function $\mathbf{d}$ enjoys properties (i'), (ii) and (iii), it is said a metric for $X$. 
When function $\mathbf{d}$ is a (quasi) (pseudo) metric for $X$, the pair $(X,\mathbf{d})$ is said a (quasi) (pseudo) metric space.
From \cite{QPM}, given a quasi--pseudo--metric space $(X,\mathbf{d})$, a sequence $\{x_{i}\}_{i\in\mathbb{N}_{0}}$ over $X$ is left (resp. right) $\mathbf{d}$--convergent to $x^{\ast}\in X$, denoted $\stackrel [\leftarrow]{}{\lim} x_{i}=x^{\ast}$ (resp. $\stackrel [\rightarrow]{}{\lim} x_{i}=x^{\ast}$), if for any $\varepsilon\in\mathbb{R}^{+}$ there exists $N\in\mathbb{N}_{0}$ such that $\mathbf{d}(x_{i},x^\ast)\leq \varepsilon$ (resp. $\mathbf{d}(x^\ast,x_{i})\leq \varepsilon$) for any $i \geq N$. 
Given $X\subseteq \mathbb{R}^{n}$ we denote by $\mathbf{d}_{h}$ the Hausdorff pseudo--metric induced by the infinity norm $\Vert \cdot \Vert$ on $2^{X}$; we recall that for any $X_{1},X_{2}\subseteq X$, \mbox{$\mathbf{d}_{h}(X_{1},X_{2}):=\max\{\vec{\mathbf{d}}_{h}(X_{1},X_{2}),\vec{\mathbf{d}}_{h}(X_{2},X_{1})\}$},
where \mbox{$\vec{\mathbf{d}}_{h}(X_{1},X_{2})=\sup_{x_{1}\in X_{1}}\inf_{x_{2}\in X_{2}} \Vert x_{1}-x_{2}\Vert$} is the Hausdorff quasi--pseudo--metric. 

\section{Piecewise Affine Systems} \label{sec3}
In this note we consider the class of discrete--time Piecewise Affine (PWA) systems described by the triplet $\Sigma=(\mathbb{R}^{n},\mathcal{U},\{\Sigma_{1},\Sigma_{2},...,\Sigma_{N}\})$, where $\mathbb{R}^{n}$ is the state space, $\mathcal{U}\subseteq \mathbb{R}^{m}$ is the set of control inputs and $\Sigma_{i}$ is a constrained affine control system defined by:
\[
\left\{
\begin{array}
{l}
x_{i}(t+1)=A_{i}x_{i}(t)+B_{i}u_{i}(t)+f_{i},\\
x_{i}(t)\in X_{i},u_{i}(t)\in \mathcal{U}.
\end{array}
\right.
\]
We suppose that the sets $X_{i}\subseteq \mathbb{R}^{n}$ are polyhedral, with interior, and that their collection is a partition of $\mathbb{R}^{n}$; moreover we suppose that the set $\mathcal{U}$ is polyhedral. We denote by $\mathbf{x}(t,x_{0},\mathbf{u})$ the state reached by $\Sigma$ at time $t\in\mathbb{N}_{0}$ starting from an initial state $x_{0}\in\mathbb{R}^{n}$ with control input $\mathbf{u}:\mathbb{N}_{0}\rightarrow \mathcal{U}$. Since $\{X_{i}\}_{i\in [1;N]}$ is a partition of $\mathbb{R}^{n}$ the PWA system $\Sigma$ is deterministic. 
In this note we are interested in the evolution of PWA systems within bounded subsets of the state space $\mathbb{R}^{n}$. This choice is motivated by the fact that in many applications, physical variables such as velocities, temperatures, pressures, voltages, take value within bounded sets. Let $\mathcal{X}$ be a polytopic subset of $\mathbb{R}^{n}$ that represents the region of the state space of $\Sigma$ which we are interested in. 
Define $\mathcal{X}_{i}=X_{i}\cap \mathcal{X}$ ($i\in [1;N]$) and denote by $\mathcal{P}(\mathcal{X})$ the set of polytopic subsets of $\mathcal{X}$.

\section{Symbolic Systems and Approximate Relations} \label{sec5}
We use the notion of systems as a unified framework to describe PWA systems as well as their symbolic models. 

\begin{definition}
\label{SysDef}
\cite{paulo} 
A system is a quintuple $S=(X,U,\rTo,Y,H)$ consisting of a set of states $X$, a set of inputs $U$, a transition relation $\rTo\subseteq X\times U\times X$, a set of outputs $Y$ and an output function $H:X\rightarrow Y$. A transition $(x,u,x^{\prime})\in\rTo$ of $S$ is denoted by $x\rTo^{u}x^{\prime}$. A state run of $S$ 
with length $T\in\mathbb{N}_{0}\cup\{\infty\}$ is a (possibly infinite) sequence of transitions $x_{0}\rTo^{u_{1}}x_{1}\rTo^{u_{2}} \,...\,\rTo^{u_{T}} x_{T}$ of $S$. An output run of $S$ with length $T\in\mathbb{N}_{0}\cup\{\infty\}$ is a (possibly infinite) sequence of output symbols $y_{0},\,y_{1},\,...\, ,y_{T}$
such that for all $y_{i}$ and $y_{i+1}$ there exists $x_{i}\rTo^{u_{i+1}}x_{i+1}$ such that $y_{i}=H(x_{i})$ and $y_{i+1}=H(x_{i+1})$. System $S$ is said to be 
\textit{symbolic}, if $X$ and $U$ are finite sets; \textit{(pseudo) metric}, if $Y$ is equipped with a (pseudo) metric $\mathbf{d}$; \textit{deterministic}, if for any state $x\in X$ and any input $u\in U$ there exists at most one transition $x\rTo^{u}x'$.
\end{definition}
In this note we use the notions of approximate simulation and bisimulation to relate properties of PWA systems and of their symbolic systems.

\begin{definition}
\label{ASR} 
\cite{AB-TAC07}
Let \mbox{$S_{1}=(X_{1},U_{1},\rTo_{1},Y_{1},H_{1})$} and \mbox{$S_{2}=(X_{2},U_{2},\rTo_{2},Y_{2},H_{2})$} be (pseudo) metric systems with the same output sets $Y_{1}=Y_{2}$ and (pseudo) metric $\mathbf{d}$ and consider a precision $\varepsilon\in\mathbb{R}^{+}_{0}$. A relation $\mathcal{R}\subseteq X_{1} \times X_{2}$ is an $\varepsilon$--approximate simulation relation from $S_{1}$ to $S_{2}$ if for every $(x_{1},x_{2})\in \mathcal{R}$ the following conditions are satisfied: 
(i) $\mathbf{d}(H_{1}(x_{1}),H_{2}(x_{2}))\leq\varepsilon$ and (ii) existence of $x_{1} \rTo_{1}^{u_{1}}x'_{1}$ implies existence of $x_{2}\rTo_{2}^{u_{2}}x'_{2}$ such that $(x^{\prime}_{1},x^{\prime}_{2})\in \mathcal{R}$. 
System $S_{1}$ is said to be $\varepsilon$--approximately simulated by $S_{2}$ or $S_{2}$ $\varepsilon$--approximately simulates $S_{1}$, denoted \mbox{$S_{1}\preceq_{\varepsilon}S_{2}$}, if $\mathcal{R}(X_{1})=X_{2}$. When $\varepsilon=0$, system $S_{1}$ is said to be exactly simulated by system $S_{2}$, or equivalently, $S_{2}$ exactly simulates $S_{1}$. 
Relation $\mathcal{R}$ is an \mbox{$\varepsilon$--approximate} bisimulation relation between $S_{1}$ and $S_{2}$ if: 
(iii) $\mathcal{R}$ is an \mbox{$\varepsilon$--approximate} simulation relation from $S_{1}$ to $S_{2}$, and (iv) $\mathcal{R}^{-1}$ is an \mbox{$\varepsilon$--approximate} simulation relation from $S_{2}$ to $S_{1}$. 
Systems $S_{1}$ and $S_{2}$ are $\varepsilon$--approximately bisimilar if $\mathcal{R}(X_{1})=X_{2}$ and $\mathcal{R}^{-1}(X_{2})=X_{1}$. If $\varepsilon=0$, $S_{1}$ and $S_{2}$ are said to be (exactly) bisimilar.
\end{definition}

In the sequel we will work with the set $\mathcal{S}(\mathcal{P}(\mathcal{X}),\mathbf{d}_{h})$ of pseudo--metric systems with output pseudo--metric space $(\mathcal{P}(\mathcal{X}),\mathbf{d}_{h})$. The notion of approximate simulation relations induces certain metrics on $\mathcal{S}(\mathcal{P}(\mathcal{X}),\mathbf{d}_{h})$. 
\begin{definition}
\cite{AB-TAC07}
Consider two pseudo--metric systems $S_{1},S_{2}\in \mathcal{S}(\mathcal{P}(\mathcal{X}),\mathbf{d}_{h})$. 
The simulation metric $\vec{\mathbf{d}}_{\s}$ from $S_{1}$ to $S_{2}$ is defined by 
$ \vec{\mathbf{d}}_{\s}(S_{1},S_{2})=\inf\{\varepsilon \in\mathbb{R}^{+}_{0} | S_{1}\preceq_{\varepsilon}S_{2} \}$. 
\end{definition}

\begin{theorem}
\cite{AB-TAC07} 
The pair $(\mathcal{S}(\mathcal{P}(\mathcal{X}),\mathbf{d}_{h}),\vec{\mathbf{d}}_{\s})$ is a quasi--pseudo--metric space\footnote{In \cite{AB-TAC07} quasi--pseudo--metric spaces are termed directed pseudo--metric spaces.}. 
\end{theorem}

\section{Sequences of Symbolic Models} \label{sec6}
The expressive power of the notion of systems as in Definition \ref{SysDef} is general enough to describe the evolution of PWA systems within the bounded region $\mathcal{X}$ of the state space $\mathbb{R}^{n}$. 

\begin{definition}
\label{SysPWA}
Given the PWA system $\Sigma$ and the polytopic subset $\mathcal{X}$ of $\mathbb{R}^{n}$ define the pseudo--metric system $\mathbb{S}(\Sigma)=(\mathbb{X},\mathbb{U},\rTo,\mathbb{Y},\mathbb{H})$, where $\mathbb{X}=\mathcal{X}$; $\mathbb{U}=\mathcal{U}$; $x \rTo^{u} x^{\prime}$, if $x\in \mathcal{X}_{i}$ and $x^{\prime}=A_{i}x+B_{i}u+f_{i}$; $\mathbb{Y}=\mathcal{P}(\mathcal{X})$, equipped with $\mathbf{d}_{h}$; $\mathbb{H}(x)=\{x\}$. 
\end{definition}

System $\mathbb{S}(\Sigma)$ preserves important properties of $\Sigma$, such as reachability and determinism. Also, since $\mathbf{d}_{h}(\{x\},\{y\})=\Vert x-y \Vert$, metric properties of $\Sigma$ are naturally transferred to $\mathbb{S}(\Sigma)$ and vice versa. Although system $\mathbb{S}(\Sigma)$ correctly describes $\Sigma$ within the bounded set $\mathcal{X}$, it is not symbolic because $\mathbb{X}$ and $\mathbb{U}$ are not finite sets. For this reason we introduce in the sequel a sequence of symbolic models $\mathbb{A}_{M}(\Sigma)$ that approximate the PWA system $\Sigma$. To this purpose we first need to introduce two operators.

\begin{definition}
Given a PWA system $\Sigma$, the bisimulation operator is the map $\Bisim:2^{\mathcal{X}}\rightarrow 2^{\mathcal{X}}$ that associates to any $Y_1,Y_2,...,Y_L\subseteq \mathcal{X}$ the collection 
$\Bisim(\{Y_1,Y_2,...,Y_L\})$ of sets $
\{
x\in Y_{j} | \exists u\in \mathcal{U}$ s.t. $A_{i}x+B_{i}u+f_{i}\in Y_{j^{\prime}},x\in \mathcal{X}_{i}
\}$ ($j,j^{\prime}\in [1;L]$).
\end{definition}

The operator $\Bisim$ transforms sets of polytopes into sets of polytopes. Note that in general, sets in $\Bisim(\{Y_1,Y_2,...,Y_L\})$ can be overlapping. 
The above definition of the bisimulation operator has been obtained by adapting  to PWA systems standard fixed point formulations of bisimulation algorithms (see e.g. \cite{ModelChecking,paulo}). A fixed point of the operator $\Bisim$, initialized with the partition $\{\mathcal{X}_{1},\mathcal{X}_{2},...,\mathcal{X}_{N}\}$ of $\Sigma$, corresponds to a finite bisimulation of $\Sigma$. Sufficient conditions for existence of finite bisimulations have been identified in \cite{LeeCDC11} for discrete-time PWA systems and in \cite{Vladimerou2008} for continuous--time PWA systems. Other 
classes of dynamical and control systems have been identified in the literature which admit finite bisimulation, as for example timed automata, multi--rate automata, rectangular automata, o--minimal hybrid systems \cite{DiscAbs} and controllable discrete--time linear systems \cite{TabuadaLTL}. 
We can now introduce the splitting operator. We recall that the diameter $\Diam(X)$ of a set $X\subseteq \mathbb{R}^{n}$ is defined by $\Diam(X)=\sup_{x,y\in X} \Vert x-y \Vert$. 

\begin{definition}
\label{DefSplit}
Consider a finite collection of polytopes $\mathbb{P}=\{P_{1},P_{2},...,P_{N}\}\subset \mathcal{P}(\mathcal{X})$. 
A splitting policy with contraction rate $\lambda\in ]0,1[$ for $\mathbb{P}$ is a map $\Phi_{\lambda}:\mathbb{P}\rightarrow 2^{\mathcal{P}(\mathcal{X})}$ 
enjoying the following properties: (i) the cardinality of $\Phi_{\lambda}(P_{i})$ is finite; (ii) $\Phi_{\lambda}(P_{i})$ is a partition of $P_{i}$; (iii) $\Diam(P_{i}^{j})\leq \lambda \Diam(P_{i})$ for all $P_{i}^{j}\in \Phi_{\lambda}(P_{i})$.
\end{definition}

In the sequel, $\Split_{\lambda}$ denotes a splitting policy with contraction rate $\lambda$ and we abuse notation by writing  $\Split_{\lambda}(\{P_{1},P_{2},...,P_{N}\})$ instead of $\bigcup_{i\in[1;N]}\Split_{\lambda}(P_{i})$. An example of splitting policy is reported in Section VII. 
The practical computation of operators $\Bisim$ and $\Split$ is based on basic manipulations of polytopes; the interested reader can refer to \cite{Belta2012} where similar computations are described in detail.\\
We now have all the ingredients to introduce a sequence of abstractions $\mathbb{A}_{M}(\Sigma)$ approximating the PWA system $\Sigma$. Consider the following recursive equations: 
\begin{equation}
\left\{
\begin{array}
{l}
\mathbf{X}_{0}=\{\mathcal{X}_{1},\mathcal{X}_{2},...,\mathcal{X}_{N}\},\\
\mathbf{X}_{M+1}=\Split_{\lambda}(\Bisim(\mathbf{X}_{M})), M\in\mathbb{N}_{0}.
\end{array}
\right.
\label{Eli}
\end{equation}

At each order $M\in \mathbb{N}_{0}$, the set $\mathbf{X}_{M}$ naturally induces a system that is formalized as follows.

\begin{definition}
\label{AbsDef}
Given the set $\mathbf{X}_{M}$ define the pseudo--metric system 
$ \mathbb{A}_{M}(\Sigma)=(\mathbb{X}_{M},\mathbb{U}_{M},\rTo_{M},$ $\mathbb{Y}_{M},\mathbb{H}_{M})$ 
where:
\begin{itemize}
\item $\mathbb{X}_{M}=\mathbf{X}_{M}\cup (\mathcal{X}\backslash \bigcup_{X\in \mathbf{X}_{M}} X )$. A state in $\mathbb{X}_{M}$ is denoted by $X_{M}^{j}$.
\item $X_{M}^{j} \rTo_{M}^{V} X_{M}^{j^{\prime}}$ if 
there exist $x\in X_{M}^{j}$ and $u\in \mathcal{U}$ such that $A_{i}x+B_{i}u+f_{i}\in  X_{M}^{j^{\prime}}$, and 
$V=\{u\in \mathcal{U} | \exists x\in X_{M}^{j} \text{ s.t. } A_{i}x+B_{i}u+f_{i}\in X_{M}^{j^{\prime}}\}$, 
where index $i$ is such that $X_{M}^{j}\subseteq \mathcal{X}_{i}$.
\item $\mathbb{U}_{M}$ is the collection of all sets $V\subseteq \mathcal{U}$ for which $X_{M}^{j} \rTo_{M}^{V} X_{M}^{j^{\prime}}$.
\item $\mathbb{Y}_{M}=\mathcal{P}(\mathcal{X})$, equipped with the pseudo--metric $\mathbf{d}_{h}$.
\item $\mathbb{H}_{M}(X_{M}^{j})=X_{M}^{j}$. 
\end{itemize}
\end{definition}

By construction, system $\mathbb{A}_{M}(\Sigma)$ is symbolic. Symbolic system $\mathbb{A}_{M+1}(\Sigma)$ can be viewed as a refinement of $\mathbb{A}_{M}(\Sigma)$. By definition of $\mathbb{X}_{M}$, the collection of sets in $\mathbb{X}_{M}$ is a covering of $\mathcal{X}$. The collection of sets in $\mathbb{U}_{M}$ instead, is in general not a covering of $\mathcal{U}$; this is because there may be control inputs that bring the state of $\Sigma$ outside the working region $\mathcal{X}$. 
The following result holds as a direct consequence of the definition of operator $\Bisim$. 

\begin{proposition}
\label{coro}
If $\mathbf{X}_{M}=\Bisim(\mathbf{X}_{M})$ then $\mathbb{A}_{M}(\Sigma)$ and $\mathbb{S}(\Sigma)$ are exactly bisimilar.
\end{proposition}

\begin{example}
Consider the PWA system $\Sigma=(\mathbb{R},\{0\},\{\Sigma_{1},\Sigma_{2}\})$, where $\Sigma_{i}$ is described by $x_{i}(t+1)=x_{i}(t)+3-2i$ and $\mathcal{X}_{i}=[-6+3i,-3+3i[$ ($i=1,2$) and let us compute the sequence of sets $\mathbf{X}_{i}$ defined in (\ref{Eli}). One first obtains $\Bisim(\mathbf{X}_{0})=\{[-3,-1[,$ $[-1,0[,[0,1[,[1,3[\}$. Note that $\Bisim(\mathbf{X}_{0})\neq \mathbf{X}_{0}$. Let $\rho=\max_{Y\in \Bisim(\mathbf{X}_{0})}\Diam(Y)=2$. Consider $\lambda=0.5$ and define for $a,b\in\mathbb{R}^{+}$ with $a<b$, $\Split_{\lambda}([a,b[)=\{[a,b[\}$, if $\Diam([a,b[)\leq \lambda \rho$ and $\Split_{\lambda}([a,b[)=\{[a,(a+b)/2[,[(a+b)/2,b[\}$, otherwise. It is easy to see that $\Split_{\lambda}$ satisfies the conditions of Definition \ref{DefSplit}. By a straightforward computation one gets 
$\mathbf{X}_{1}=\Split_{\lambda}(\Bisim(\mathbf{X}_{0}))=\{[-3,-2[,[-2,-1[,[-1,0[,[0,1[,$ $[1,2[,$ $[2,3[\}$ 
from which, $\Bisim(\mathbf{X}_{1})=\mathbf{X}_{1}$. 
Hence, from Proposition \ref{coro}, $\mathbb{A}_{1}(\Sigma)$ is an exact bisimulation of $\Sigma$.
\end{example}

We point out that in general, even if an exact bisimulation of a given PWA system $\Sigma$ exists, there is no guarantee it can be found by the recursive equations in (\ref{Eli}); this is because in general $\Split_{\lambda}$ does not satisfy the reachability properties of $\Sigma$. 
On the other hand, as we shall show in the sequel, the splitting operator is a key element to prove the convergence properties of the sequence $\mathbb{A}_{M}(\Sigma)$. 
We now proceed with a step further by providing a quantification of the accuracy of the approximation scheme that we propose. 
Define $\Gran(\mathbb{A}_{M}(\Sigma))=\max_{X_{M}^{j}\in \mathbb{X}_{M}}\Diam(X_{M}^{j})$. 
Function $\Gran$ provides a measure of the "granularity" of the symbolic system $\mathbb{A}_{M}(\Sigma)$ (i.e. how fine is the covering of the set $\mathcal{X}$). The following result provides an upper bound to the distance between the PWA system $\Sigma$ and the abstraction $\mathbb{A}_{M}(\Sigma)$. 

\begin{theorem}
\label{ThMain2}
$\vec{\mathbf{d}}_{\s}(\mathbb{S}(\Sigma),\mathbb{A}_{M}(\Sigma))\leq \Gran(\mathbb{A}_{M}(\Sigma))$. 
\end{theorem}
\begin{proof}
Define $\mathcal{R}\subseteq \mathbb{X}\times \mathbb{X}_{M}$ such that $(x,X_{M}^{j})\in\mathcal{R}$ if and only if $x\in X_{M}^{j}$. Consider any $(x,X_{M}^{j})\in\mathcal{R}$. 
By definition of $\Gran(\mathbb{A}_{M}(\Sigma))$ one gets $
\mathbf{d}_{h}(\mathbb{H}(x),\mathbb{H}_{M}(X_{M}^{j}))$ $\leq \Diam(X_{M}^{j}) \leq \Gran(\mathbb{A}_{M}(\Sigma))$. 
Hence, condition (i) in Definition \ref{ASR} is satisfied. Consider any transition $x\rTo^{u} x^{\prime}$ in $\mathbb{S}(\Sigma)$. 
By definition of $\mathbb{A}_{M}(\Sigma)$ there exists a transition $X_{M}^{j}\rTo^{V}_{M} X_{M}^{j^{\prime}}$ with $u\in V$ and $x^{\prime}\in X_{M}^{j^{\prime}}$, or equivalently $(x^{\prime},X_{M}^{j^{\prime}})\in\mathcal{R}$. Hence, condition (ii) in Definition \ref{ASR} holds.
Since $\mathbb{X}_{M}$ is a covering of $\mathcal{X}$ then $\mathcal{R}(\mathbb{X})=\mathbb{X}_{M}$ from which, condition (iii) holds. 
Hence, $\mathbb{S}(\Sigma)\preceq_{\varepsilon} \mathbb{A}_{M}(\Sigma)$ with $\varepsilon=\Gran(\mathbb{A}_{M}(\Sigma))$. Finally, the result follows from the definition of $\vec{\mathbf{d}}_{\s}$. 
\end{proof}

The rest of this section is devoted to study the convergence of the sequence $\{\mathbb{A}_{M}(\Sigma)\}_{M\in\mathbb{N}_{0}}$ to $\mathbb{S}(\Sigma)$. We start by presenting the following technical result. 

\begin{lemma}
\label{lemma}
$\Gran(\mathbb{A}_{M+1}(\Sigma))\leq \lambda\Gran(\mathbb{A}_{M}(\Sigma))$.
\end{lemma}

\begin{proof}
By (\ref{Eli}) for all states $X_{M+1}^{j}$ in $\mathbb{A}_{M+1}(\Sigma)$ there exist a state $X_{M}^{i_j}\in \mathbb{X}_{M}$ such that $X_{M+1}^{j}\subseteq \Split_{\lambda}(Z)$ for some set $Z\subseteq X_{M}^{i_j}$. By the above condition and the definition of $\Split_{\lambda}$, the inequality $\Diam (X_{M+1}^{j}) \leq \lambda \Diam(Z) \leq \lambda \Diam(X_{M}^{i_j})$ holds. Hence, by applying the operator $\max$ to both sides of the above inequality, one gets 
$\Gran(\mathbb{A}_{M+1}(\Sigma)) = \max_{j} \Diam (X_{M+1}^{j}) \leq \lambda \max_{j} \Diam (X_{M}^{i_j}) \leq \lambda \max_{i} \Diam (X_{M}^{i})=  \lambda\Gran(\mathbb{A}_{M}(\Sigma))$, which concludes the proof. 
\end{proof}

We now have all the ingredients to present one of the main results of this note.

\begin{theorem}
\label{ThMain3}
$ \mathbb{S}(\Sigma)=\stackrel [\leftarrow]{}{\lim} \mathbb{A}_{M}(\Sigma) $.
\end{theorem}

\begin{proof}
Pick any $\varepsilon \in\mathbb{R}^{+}$ and choose $M_{\varepsilon}\in\mathbb{N}_{0}$ such that $\lambda^{M_{\varepsilon}-1}\Gran(\mathbb{A}_{1}(\Sigma))\leq \varepsilon$ .
By the above inequality and by combining Theorem \ref{ThMain2} and Lemma \ref{lemma} for all $M\geq M_{\varepsilon}$, we obtain 
$\vec{\mathbf{d}}_{\s}(\mathbb{S}(\Sigma),\mathbb{A}_{M}(\Sigma)) \leq \Gran(\mathbb{A}_{M}(\Sigma)) \leq \lambda \Gran(\mathbb{A}_{M-1}(\Sigma)) \leq ... \leq  \lambda^{M-1}\Gran(\mathbb{A}_{1}(\Sigma)) \leq \lambda^{M_{\varepsilon}-1}\Gran(\mathbb{A}_{1}(\Sigma))\leq \varepsilon 
$, which concludes the proof.
\end{proof}

\section{Symbolic Control Design} \label{sec7}

In this section we address the design of symbolic control for PWA systems where specifications are expressed in terms of non-deterministic finite automata. This class of specifications is rather general and comprises for example a fragment of Linear Temporal Logic (LTL) formulae called syntactically co--safe LTL formulae. Syntactically co-safe LTL formulae include a large spectrum of finite--time specifications as for example reachability problems with obstacle avoidance and enabling conditions (see e.g. \cite{csLTL} for further details). 
Consider a specification described by the pseudo--metric symbolic system $Q=(\mathbb{X}^{q},\mathbb{U}_{q},\rTo_{q},\mathbb{Y}_{q},\mathbb{H}_{q})$, where 
$\mathbb{X}^{q}=\{\mathcal{X}^{q}_{1},\mathcal{X}^{q}_{2},...,\mathcal{X}^{q}_{N^{q}}\}$ is a finite collection of polytopic subsets of $\mathcal{X}$; 
$\mathbb{U}_{q}=\{0\}$; 
$\rTo_{q}\subseteq \mathbb{X}^{q}\times \mathbb{U}^{q}\times \mathbb{X}^{q}$; 
$\mathbb{Y}_{q}=2^\mathcal{X}$, equipped with the pseudo--metric $\mathbf{d}_{h}$; 
$\mathbb{H}_{q}(\mathcal{X}^{q}_{i})=\mathcal{X}^{q}_{i}$. 
Define $\mathcal{X}^{q}=\bigcup_{i\in[1;N^{q}]}\mathcal{X}^{q}_{i}$. 
We suppose that the collection $\mathbb{Y}_{q}$ of sets $\mathcal{X}^{q}_{i}$ is contained in the partition $\{\mathcal{X}_{i}\}_{i\in [1;N]}$ of $\mathcal{X}$; this assumption can be given without loss of generality by appropriately duplicating the dynamics of $\Sigma$. For ease of notation we denote in the sequel a transition $\mathcal{X}^{q}_{i_1}\rTo_{q}^{0}\mathcal{X}^{q}_{i_2}$ by $\mathcal{X}^{q}_{i_1}\rTo_{q}\mathcal{X}^{q}_{i_2}$. 
The class of control strategies that we consider in this note is specified by a partition $\mathbf{P}=\{\mathbf{P}_{i}\}_{i\in I}$ of $\mathcal{X}$ and a map $\mathcal{K}: \mathbf{P} \rightarrow 2^{\mathcal{U}}$. 
Note that we are not supposing that $\mathbf{P}$ is either finite or countable. When $\mathbf{P}$ is a finite set, the control strategy is said symbolic. 
Map $\mathcal{K}$ associates to an aggregate of states $\mathbf{P}_{i} \in \mathbf{P}$ an aggregate of inputs $\mathcal{K}(\mathbf{P}_{i})\subseteq \mathcal{U}$ representing the collection of admissible inputs.  
Given a control strategy $\mathcal{K}$, we denote by $\Sigma^{\mathcal{K}}$ the closed--loop PWA system $\Sigma$ where $u=\kappa(x)\in \mathcal{K}(\mathbf{P}_{i})$ if $x\in \mathbf{P}_{i}$. With abuse of notation, we denote by $\mathbf{x}(t,x_{0},\kappa)$ the state reached by $\Sigma$ at time $t$ starting from an initial state $x_{0}\in\mathcal{X}$ with feedback control law $\kappa(x)\in\mathcal{K}(\mathbf{P}_{i})$, $x\in \mathbf{P}_{i}$; moreover we write $\kappa\in\mathcal{K}$ when $\kappa(x)\in \mathcal{K}(\mathbf{P}_{i})$ for all $x\in \mathbf{P}_{i}$ and $\mathcal{K}\subseteq \mathcal{K}^{\prime}$ when $\mathcal{K}(\{x\})\subseteq \mathcal{K}^{\prime}(\{x\})$ for all $x\in \mathcal{X}$. We can now formally state the control design problem considered in this note. 
\begin{definition}
A control strategy $\mathcal{K}:\mathbf{P}\rightarrow 2^{\mathcal{U}}$ is said to enforce the specification $Q$ 
on $\Sigma$ if for all initial states $x_{0}\in \mathbf{P}_{i}$ of $\Sigma$ for which $\mathcal{K}(\mathbf{P}_{i})\neq \varnothing$ and for all $\kappa\in \mathcal{K}$ there exists a (possibly infinite) state run $\mathcal{X}^{q}_{i_0}\rTo_{q}\mathcal{X}^{q}_{i_2}\rTo_{q} \,...\,\rTo_{q} \mathcal{X}^{q}_{i_T}$ of $Q$ with length $T$ such that $\mathbf{x}(t,x_{0},\kappa)\in \mathcal{X}^{q}_{i_t}$ and $\mathbf{x}(t+1,x_{0},\kappa)\in \mathcal{X}^{q}_{i_{t+1}}$ for all $t\in [0;T-1]$.
\end{definition}
In the above definition, a control strategy enforces the specification $Q$ in the sense of the so--called similarity games, see e.g. \cite{paulo}. Existence of such a control strategy guarantees that for all initial states for which the control strategy is not empty, the corresponding state runs satisfy the specification. This definition does not exclude the trivial case where the set of initial states, for which a control strategy enforces a given specification, is empty. However, in the sequel we will be interested in (approximating) the maximal control strategy (in the sense of Definition \ref{Cmax}). Hence in that case, if the set of states for which the maximal controller is empty, the control problem has no solution. 
Denote by $\mathbf{K}(\Sigma,Q)$ the collection of all control strategies enforcing the specification $Q$ on $\Sigma$. 
\begin{definition}
\label{Cmax}
The \textit{maximal control strategy} enforcing the specification $Q$ on the PWA system $\Sigma$, is a control strategy $\mathcal{K}^{\ast}\in \mathbf{K}(\Sigma,Q)$ such that $\mathcal{K}\subseteq \mathcal{K}^{\ast}$ for all $\mathcal{K}\in \mathbf{K}(\Sigma,Q)$.
\end{definition}
\begin{proposition}
$\mathcal{K}^{\ast}(\{x\})=\bigcup_{\mathcal{K}\in \mathbf{K}(\Sigma,Q)} \mathcal{K}(\{x\})$.
\end{proposition}

From the above result the control strategy $\mathcal{K}^{\ast}$ exists and is unique. In general, control strategy $\mathcal{K}^{\ast}$ is not symbolic and its explicit expression cannot be easily derived. For this reason in the sequel we propose a sequence of control strategies $\mathcal{K}_{M}$, approximating $\mathcal{K}^{\ast}$, that can be computed on the basis of the symbolic systems $\mathbb{A}_{M}(\Sigma)$.  

\begin{definition}
\label{SymbContrDef}
Given the system $\mathbb{A}_{M}(\Sigma)$, define for all $X_{M}^{j}\in\mathbb{X}_{M}$ the graph $\mathcal{G}(X_{M}^{j})=(\mathcal{N},\mathcal{E})$ where $\mathcal{N}$ is the collection of sets $V\in\mathbb{U}_{M}$ such that $X_{M}^{j} \rTo_{M}^{V} X_{M}^{j^{\prime}}$ and $\mathcal{E}$ is the collection of all pairs $(V,V^{\prime})\in \mathcal{N}\times \mathcal{N}$ such that $V\cap V^{\prime}\neq \varnothing$. 
For all connected components $\mathbf{G}_{i}(X_{M}^{j})$ of $\mathcal{G}(X_{M}^{j})$ define the following sets:
$\mathbf{U}_{i}(X_{M}^{j})$ is the union of nodes of $\mathbf{G}_{i}(X_{M}^{j})$; 
$\mathbf{P}_{i}(X_{M}^{j})$ is the union of sets $X_{M}^{j^{\prime}}$ for which $X_{M}^{j} \rTo_{M}^{V}X_{M}^{j^{\prime}}$ and $V$ is a node of $\mathbf{G}_{i}(X_{M}^{j})$. 
Define the control strategy $\mathcal{K}_{M}:\mathbb{X}_{M}\rightarrow 2^{\mathcal{U}}$ such that: 
for all $X_{M}^{j}\not\subseteq \mathcal{X}^{q}$, $\mathcal{K}_{M}(X_{M}^{j})=\varnothing$; for all $X_{M}^{j}\subseteq \mathcal{X}^{q}$, $\mathcal{K}_{M}(X_{M}^{j})=\bigcup_{i} \mathbf{U}_{i}(X_{M}^{j}) \text{ s.t. } \mathbf{P}_{i}(X_{M}^{j})\subseteq \Post_{q}(X_{M}^{j})$, where $\Post_{q}(X_{M}^{j})$ is the union of sets $\mathcal{X}^{q}_{j^{\prime}}\in \mathbb{X}^{q}$ such that 
$X_{M}^{j}\subseteq \mathcal{X}^{q}_{j}\in\mathbb{X}^{q}$ and $\mathcal{X}^{q}_{j}\rTo_{q} \mathcal{X}^{q}_{j^{\prime}}$. 
\end{definition}

From the above definition it is easy to see that $\mathcal{K}_{M}$ is symbolic. 
Moreover, $\mathcal{K}_{M}$ guarantees that the closed--loop PWA system $\Sigma^{\mathcal{K}_{M}}$ satisfies the specification $Q$, as formally stated in the following result. 

\begin{theorem}
$\mathcal{K}_{M} \in \mathbf{K}(\Sigma,Q)$. 
\end{theorem}

\begin{proof}
We prove the statement by induction, by showing that starting from a state $x\in \mathcal{X}^{q}$ fulfilling the specification, by applying a control strategy $\mathcal{K}_{M} \in \mathbf{K}(\Sigma,Q)$ a state is reached which again satisfies the specification. Consider any $x\in \mathcal{X}^{q}$ for which $\mathcal{K}_{M}(\{x\})\neq\varnothing$ and any $u\in \mathcal{K}_{M}(\{x\})$. Let $X_{M}^{j}\in\mathbb{X}_{M}$ be such that $x\in X_{M}^{j}$. 
Since $u\in \mathcal{K}_{M}(\{x\})$ there exists a connected component $\mathbf{G}_{i}(X_{M}^{j})$ of $\mathcal{G}(X_{M}^{j})$ such that $u\in \mathbf{U}_{i}(X_{M}^{j})$ and $\mathbf{x}(x,1,u) \in \mathbf{P}_{i}(X_{M}^{j})$. Since $\mathbf{P}_{i}(X_{M}^{j})\subseteq \Post_{q}(X_{M}^{j})$ the specification $Q$ is satisfied. 
\end{proof}

The practical computation of the symbolic controller $\mathcal{K}_{M}$ is based on basic operations on graphs and polytopes. 
The following result establishes a sufficient condition to find the maximal control strategy $\mathcal{K}^{\ast}$.

\begin{proposition}
If $\mathbf{X}_{M}=\Bisim(\mathbf{X}_{M})$ then $\mathcal{K}^{\ast}=\mathcal{K}_{M}$.
\end{proposition}

The proof of the above result is a direct 
consequence of the definitions of $\mathcal{K}_{M}$ and $\mathcal{K}^{\ast}$ and of Proposition \ref{coro} and is therefore omitted. We conclude this section by showing that the sequence $\mathcal{K}_{M}$ converges to $\mathcal{K}^{\ast}$. We firstly provide a representation of (symbolic) control strategies in terms of (symbolic) systems. 

\begin{definition}
Given the control strategy $\mathcal{K}^{\ast}$ define the pseudo--metric system $\mathbb{S}(\mathcal{K}^{\ast})=(\mathbb{X},\mathbb{U},$ $\rTo_{\mathcal{K}^{\ast}},\mathbb{Y},\mathbb{H})$, 
where entities $\mathbb{X}$, $\mathbb{U}$, $\mathbb{Y}$ and $\mathbb{H}$ are defined in Definition \ref{SysPWA} and 
 $x\rTo_{\mathcal{K}^{\ast}}^{u} x^{\prime}$ if and only if $x\rTo^{u} x^{\prime}$ in $\mathbb{S}(\Sigma)$ and $u\in \mathcal{K}^{\ast}(\{x\})$. 
\end{definition}

\begin{definition}
Given the symbolic control strategy $\mathcal{K}_{M}$ define the pseudo--metric symbolic system $\mathbb{S}(\mathcal{K}_{M})=(\mathbb{X}_{M},\mathbb{U}_{M},\rTo_{\mathcal{K}_{M}},\mathbb{Y}_{M},\mathbb{H}_{M})$, where entities $\mathbb{X}_{M}$, $\mathbb{U}_{M}$, $\mathbb{Y}_{M}$ and $\mathbb{H}_{M}$ are defined in Definition \ref{AbsDef} and $X_{M}^{j}\rTo_{\mathcal{K}_{M}}^{V}X_{M}^{j^{\prime}}$ if and only if $X_{M}^{j}\rTo^{V}_{M} X_{M}^{j^{\prime}}$ in $\mathbb{A}_{M}(\Sigma)$ and $V\subseteq \mathcal{K}_{M}(X_{M}^{j})$. 
\end{definition}

We can now give the following result that quantifies the distance between $\mathcal{K}_{M}$ and $\mathcal{K}^{\ast}$. 

\begin{theorem}
\label{ThMain4}
$\vec{\mathbf{d}}_{\s}(\mathbb{S}(\mathcal{K}_{M}),\mathbb{S}(\mathcal{K}^{\ast}))\leq \Gran(\mathbb{A}_{M}(\Sigma))$. 
\end{theorem}

\begin{proof}
Define $\mathcal{R}\subseteq \mathbb{X}_{M}\times \mathbb{X}$ such that $(X_{M}^{j},x)\in\mathcal{R}$ if and only if $x\in X_{M}^{j}$. Consider any $(X_{M}^{j},x)\in\mathcal{R}$. By definition of $\Gran(\mathbb{A}_{M}(\Sigma))$ one gets $\mathbf{d}_{h}(\mathbb{H}_{M}(X_{M}^{j}),\mathbb{H}(x)) \leq \Diam(\mathbb{H}_{M}(X_{M}^{j})) \leq \Gran(\mathbb{A}_{M}(\Sigma))$ from which, condition (i) in Definition \ref{ASR} holds. We now show that also condition (ii) in Definition \ref{ASR} is satisfied. Consider any transition $X_{M}^{j}\rTo^{V}_{\mathcal{K}_{M}} X_{M}^{j^{\prime}}$ in $\mathbb{S}(\mathcal{K}_{M})$. By definition of 
$\mathcal{K}^{\ast}$, $\mathcal{K}_{M}(X_{M}^{j})\subseteq \mathcal{K}^{\ast}(X_{M}^{j})$. Hence, for all $u\in\mathcal{K}_{M}(X_{M}^{j^{\prime}})$, $x \rTo_{\mathcal{K^{\ast}}}^{u} x^{\prime}$. In particular, by definition of $\mathcal{K}_{M}$ there exists $u\in \mathcal{K}_{M}(X_{M}^{j})\subseteq \mathcal{K}^{\ast}(X_{M}^{j})$ such that $x \rTo_{\mathcal{K^{\ast}}}^{u} x^{\prime}$ and $x^{\prime}\in X_{M}^{j^{\prime}}$ from which, condition (ii) in Definition \ref{ASR} is satisfied. Since $\mathbb{X}_{M}$ is a partition of $\mathcal{X}$ then condition (iii) in Definition \ref{ASR} holds. Finally, the result follows from the definition of $\vec{\mathbf{d}}_{\s}$. 
\end{proof}

We can now present the second main result of this note.

\begin{theorem}
\label{ThMain5}
$\stackrel [\rightarrow]{}{\lim} \mathbb{S}(\mathcal{K}_{M})=\mathbb{S}(\mathcal{K}^{\ast})$.
\end{theorem}

The proof of the above result can be obtained by combining Lemma \ref{lemma} and Theorem \ref{ThMain4}, along the lines of the proof of Theorem \ref{ThMain3}, and is therefore omitted. 

\section{An illustrative example}

Consider a PWA system $\Sigma=(\mathbb{R}^{2},\mathcal{U},\{\Sigma_{1},\Sigma_{2},\Sigma_{3},\Sigma_{4}\})$, where:
\[
\begin{array}
{l}
A_{1}=\left[
\begin{array}
{rr}
0.5  & 0 \\
0 & -0.5
\end{array}
\right];
B_{1}=\left[
\begin{array}
{rr}
0  & 0 \\
0 &  0
\end{array}
\right];
f_{1}=
\left[
\begin{array}
{r}
3 \\
0.7
\end{array}
\right]; 
\\
A_{2}=\left[
\begin{array}
{rr}
0.3  & 0.1 \\
0 & 0.2
\end{array}
\right];
B_{2}=\left[
\begin{array}
{rr}
0.1  & 0 \\
0 & 0.2
\end{array}
\right];
f_{2}=
\left[
\begin{array}
{r}
0 \\
0.4
\end{array}
\right]; 
\\
A_{3}=\left[
\begin{array}
{rr}
0.8 & 0 \\
0.2 & 0.2
\end{array}
\right];
B_{3}=\left[
\begin{array}
{rr}
0.1  & 0 \\
0 & 0.2
\end{array}
\right];
f_{3}=
\left[
\begin{array}
{r}
-2.5 \\
0
\end{array}
\right]; 
\\
A_{4}=\left[
\begin{array}
{rr}
0.2  & 0 \\
0 & 0.2
\end{array}
\right];  
B_{4}=\left[
\begin{array}
{rr}
0.5  & 0 \\
0 & 0.5
\end{array}
\right];
f_{4}=
\left[
\begin{array}
{r}
1.7 \\
1
\end{array}
\right]. 
\end{array}
\]
We set $\mathcal{U}=[-0.25,0.25]\times[-0.25,0.25]$, $\mathcal{X}_{1}=[-3,-1[\times [0,2[$, $\mathcal{X}_{2}=[-1,1[\times [0,2[$, $\mathcal{X}_{3}=[1,3[\times [0,2[$ and $\mathcal{X}_{4}=[-3,3[\times [-2,0[$. 
\begin{figure*}[t]
\label{fig01}
\centering
\subfigure[ ]{
\begin{tikzpicture}[->,
shorten >=0.1pt,%
auto,node distance=1.5cm,
inner sep=0.1 pt ,bend angle=30]
\tikzstyle{every state}=[minimum size=6mm]
\tikzstyle{every node}=[font=\footnotesize]
\node[state, fill=blue!60] (1) {$\mathcal{X}_{1}$};
\node[state, fill=green!80] (2) [right of=1] {$\mathcal{X}_{2}$};
\node[state, fill=yellow!80] (3) [right of=2] {$\mathcal{X}_{3}$};
\node[state, fill=red!80] (4) [below of=2] {$\mathcal{X}_{4}$};
\path[->] 
(1) edge node {} (4)
(1) edge [bend left] node {} (3)
(2) edge [loop left] node {} (2)
(3) edge node {} (2)
(3) edge [bend left] node {} (1)
(4) edge node {} (3)
(4) edge node {} (2);
\end{tikzpicture}
} 
\subfigure[ ]{
\begin{tikzpicture}[->,
shorten >=0.1pt,%
auto,node distance=1.5cm, 
inner sep=0.1 pt ,bend angle=20]
\tikzstyle{every state}=[minimum size=6mm]
\tikzstyle{every node}=[font=\footnotesize]
\node[state, fill=blue!60] (1) {$X_{1}^{1}$};
\node[state, fill=blue!60] (2) [ right of=1] {$X_{1}^{2}$};
\node[state, fill=blue!60] (3) [ right of=2] {$X_{1}^{3}$};
\node[state, fill=blue!60] (4) [ right of=3] {$X_{1}^{4}$};
\node[state, fill=green!80] (5) [ right of=4] {$X_{1}^{5}$};
\node[state, fill=green!80] (6) [ right of=5] {$X_{1}^{6}$};
\node[state, fill=yellow!80] (7) [ right of=6] {$X_{1}^{7}$};
\node[state, fill=yellow!80] (8) [ below of=1] {$X_{1}^{8}$};
\node[state, fill=yellow!80] (9) [ right of=8] {$X_{1}^{9}$};
\node[state, fill=yellow!80] (10) [ right of=9] {$X_{1}^{10}$};
\node[state, fill=red!80] (11) [ right of=10] {$X_{1}^{11}$};
\node[state, fill=red!80] (12) [ right of=11] {$X_{1}^{12}$};
\node[state, fill=red!80] (13) [ right of=12] {$X_{1}^{13}$};
\node[state, fill=red!80] (14) [ right of=13] {$X_{1}^{14}$};
\path[->] 
(1) edge [bend right] node {} (7)
(1) edge node {} (9)
(2) edge node {} (9)
(3) edge node {} (14)
(4) edge node {} (14)
(5) edge [loop left] node {} (5)
(6) edge node {} (5)
(7) edge [bend left] node {} (2)
(7) edge [bend left] node {} (5)
(8) edge node {} (2)
(8) edge node {} (5)
(9) edge node {} (2)
(9) edge node {} (5)
(10) edge node {} (2)
(10) edge node {} (5)
(10) edge node {} (6)
(11) edge node {} (7)
(11) edge node {} (5)
(12) edge node {} (5)
(12) edge node {} (6)
(12) edge node {} (7)
(12) edge [bend right] node {} (8)
(13) edge node {} (5)
(13) edge node {} (6)
(13) edge node {} (7)
(13) edge [bend right] node {} (8)
(14) edge node {} (7)
(14) edge [bend right] node {} (8)
(14) edge [bend right] node {} (9)
(14) edge [bend right] node {} (10);
\end{tikzpicture}
}
\label{fig0}
\caption{In the left panel, the system induced by $\Bisim(\mathbf{X}_{0})$. In the right panel, system $\mathbb{A}_{1}(\Sigma)$ induced by $\mathbf{X}_{1}=\Split_{\lambda}(\Bisim(\mathbf{X}_{0}))$. The colors in the two systems indicate which state of the system in the right panel is a refinement of a state of the system in the left panel. For example states $X_{1}^{5}$ and $X_{1}^{6}$ in the right panel are marked green because they are a refinement of the green state $\mathcal{X}_{2}$ in the left panel.}
\end{figure*}
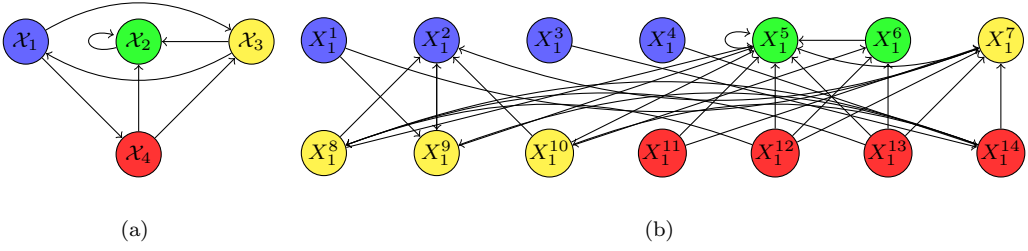
We choose a contraction rate $\lambda=0.8$ and we consider the following splitting policy. Given a polytopic set $P$ let $[a_{1},b_{1}]\times [a_{2},b_{2}]\times ... \times[a_{n},b_{n}]$ be the smallest hyperrectangle containing $P$ and let $i^{\ast}=\arg \max_{i\in [1;n]}\Vert b_{i}-a_{i} \Vert$. Define $\Split_{\lambda}(P)=\{P_1,P_2\}$, where 
$P_1=P\cap ([a_1,b_1] \times ... \times [a_{i^{\ast}-1},b_{i^{\ast}-1}] \times [a_{i^{\ast}},(a_{i^{\ast}}+b_{i^{\ast}})/2]  \times [a_{i^{\ast}+1},b_{i^{\ast}+1}] \times ... \times [a_n,b_n])$ and $P_2=P\cap ([a_1,b_1] \times ... \times [a_{i^{\ast}-1},b_{i^{\ast}-1}] \times [(a_{i^{\ast}}+b_{i^{\ast}})/2,b_{i^{\ast}}]  \times [a_{i^{\ast}+1},b_{i^{\ast}+1}] \times ... \times [a_n,b_n])$. 
If $\Diam(P_{i})> \lambda \Diam(P)$ the above splitting policy is repeated, until the sets obtained satisfy condition (iii) in Definition \ref{DefSplit}. We computed the symbolic systems $\mathbb{A}_{M}(\Sigma)$ with orders $M=1,2,3,4$. Figure 1 shows the construction of system $\mathbb{A}_{1}(\Sigma)$: in the left panel the system induced by $\Bisim(\mathbf{X}_{0})$ and in the right panel system $\mathbb{A}_{1}(\Sigma)$. The operator $\Split_{\lambda}$ cuts set $\mathcal{X}_{1}$ (resp. $\mathcal{X}_{2}$; $\mathcal{X}_{3}$; $\mathcal{X}_{4}$) into sets $X_{1}^{i}$ ($i\in [1;4]$) (resp.
($i\in [5;6]$); ($i\in [7;10]$); ($i\in [11;14]$)). The transition relation of the system in Figure 1 (a) induces the transition relation of the system in Figure 1 (b); for example, transition $\mathcal{X}_{1} \rTo \mathcal{X}_{4}$ in Figure 1 (a) corresponds to the two transitions $X_{1}^{3} \rTo X_{1}^{14}$ and $X_{1}^{4} \rTo X_{1}^{14}$ in Figure 1 (b). We do not report details on $\mathbb{A}_{i}(\Sigma)$, $i\in [2;4]$ for lack of space. Table \ref{TabEx} details space and time complexity indicators in constructing $\mathbb{A}_{i}(\Sigma)$ and the granularity indicator $\Gran(\mathbb{A}_{M}(\Sigma))$. Function $\Gran(\mathbb{A}_{M}(\Sigma))$ is decreasing and such that $\Gran(\mathbb{A}_{M+1}(\Sigma))\leq \lambda \Gran(\mathbb{A}_{M}(\Sigma))$. 
We now use these symbolic systems to solve a control design problem. Our specification $Q$ consists in a finite--time reachability problem with obstacle avoidance and time constraints: starting from region $\mathcal{X}_{1}$, reach region $\mathcal{X}_{3}$ in at most two steps while avoiding region $\mathcal{X}_{2}$ to then return in one step to region $\mathcal{X}_{1}$. This specification translates in the collection of transitions 
$\mathcal{X}_1 \rTo_{q} \mathcal{X}_1 \rTo_{q} \mathcal{X}_3 \rTo_{q} \mathcal{X}_1$, $\mathcal{X}_{1} \rTo_{q} \mathcal{X}_{3}\rTo_{q} \mathcal{X}_{1}$
and $\mathcal{X}_{1}\rTo_{q} \mathcal{X}_{4}\rTo_{q} \mathcal{X}_{3}\rTo_{q} \mathcal{X}_{1}$. 
We implemented the results presented in the previous section and we obtained the controller $\mathcal{K}_{M}$. Figure 2 illustrates for any order $M\in [1;4]$, the collection of states $x\in\mathcal{X}$ for which $\mathcal{K}_{M}(x) \neq \varnothing$. 
It is readily seen that this collection is increasing (in the sense of the preorder induced by $\subseteq$) with respect to $M$: as soon as the abstraction becomes finer the corresponding controller is able to find larger regions of the state space which satisfy the specification. 
Table \ref{TabEx} (last column) reports the percentage of the area of the region $\mathcal{X}$ that is covered by a non-empty controller solving the specification. 
Note that for $M=1$ and $M=2$ there is no control strategy that steers a state of region $\mathcal{X}_{3}$ into region $\mathcal{X}_{1}$ (set $\mathcal{X}_{3}=[1,3[\times [0,2[$ is covered by no coloured polytope). 

\begin{figure*}[t]
\label{fig}
\begin{center}
\includegraphics[scale=0.2]{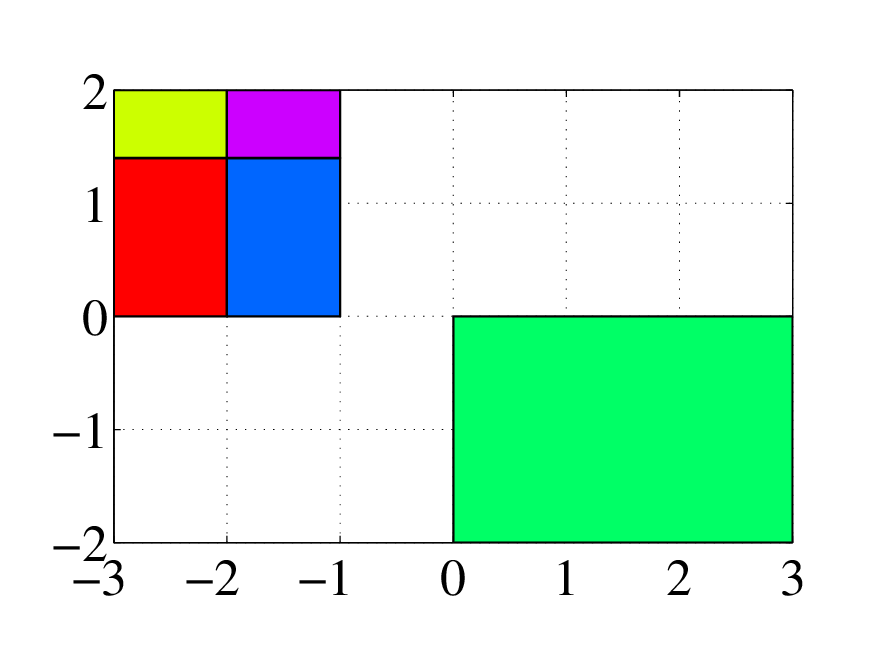}
\includegraphics[scale=0.2]{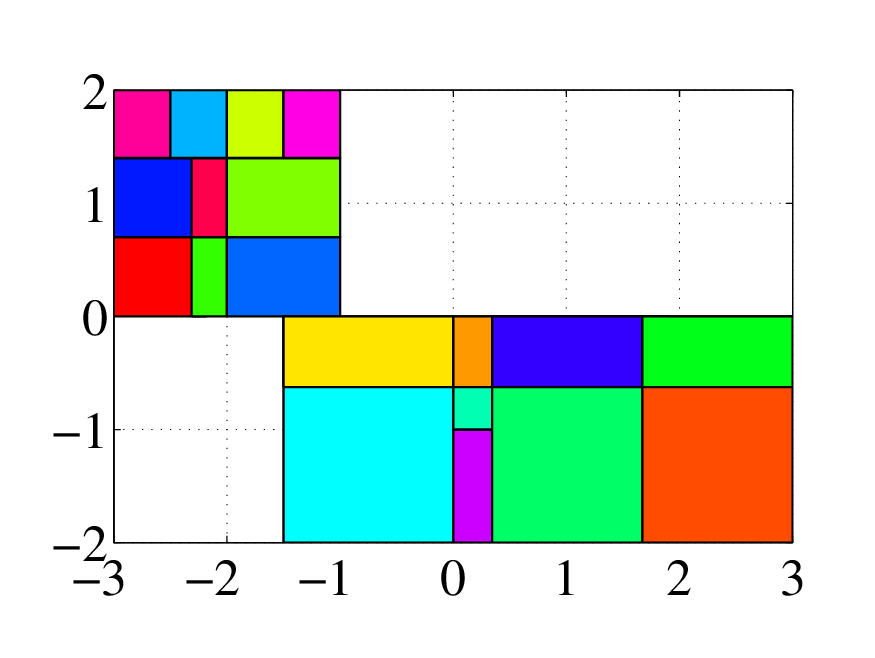}
\includegraphics[scale=0.2]{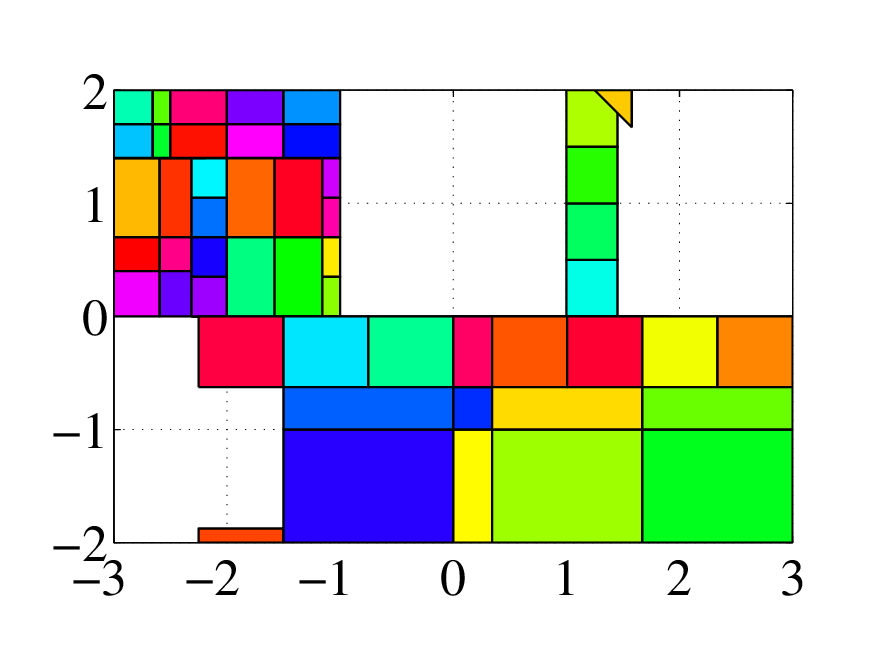}
\includegraphics[scale=0.2]{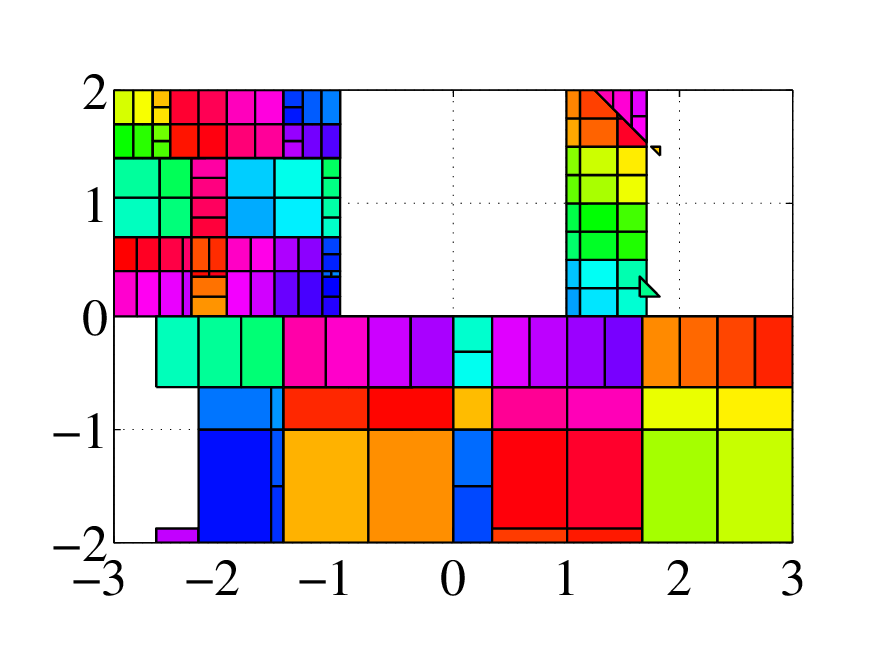}
\caption{
From the left, collection of states $x\in\mathcal{X}$ for which $\mathcal{K}_{M}(x)\neq \varnothing$, for $M=1,2,3,4$. 
}
\end{center}
\end{figure*}
\begin{table}
\begin{center}
\begin{tabular}
[c]{||l|r|r|r|r||r||}
\hline
\hline
M & $|\mathbb{X}_{M}|$ & $|\rTo_{M}|$ & $\Gran(\mathbb{A}_{M}(\Sigma))$ & Time (s) & \% \\
\hline
1 & 14 & 30 & 3 & 0.3642 & 41.66\\
\hline
2 & 48 & 166 & 2 & 2.3612 & 54.16\\
\hline
3 & 172 & 814 & 1.5 & 28.9883 & 62.60 \\
\hline
4 & 564 & 4847 & 1 & 408.3723 & 73.80\\
\hline
\hline
\end{tabular}
\caption{ }
\label{TabEx}
\end{center}
\end{table}

\section{Discussion} \label{sec8}
In this note we proposed an approach based on the notion of approximate simulation to the construction of symbolic models and the control design of PWA systems.
If compared to previous work on discrete abstractions of PWA systems, while \cite{Lee2012} and \cite{BeltaTACTN2010} use a sequence of "simulations" for
stability and formal verification problems, respectively, this work uses a sequence of "approximate simulations" for the design of symbolic controllers that satisfy a symbolic specification within prescribed accuracy.
Approximate (bi)simulation has been also employed in \cite{PolaAutom2008,MajidTAC11} and \cite{GirardTAC2010} for the construction of symbolic models for nonlinear control and switched systems. Our results compare as follows, to these works. The results of \cite{PolaAutom2008,GirardTAC2010} propose approximately bisimilar symbolic models for incrementally stable nonlinear control and switched systems. 
Our results are weaker than the ones in \cite{PolaAutom2008,GirardTAC2010} (approximate simulation vs. approximate bisimulation) but do not require stability of PWA systems. Moreover, the results in \cite{PolaAutom2008,GirardTAC2010} cannot be directly applied to the present framework because PWA systems are characterized by state--dependent discrete transitions. The work in \cite{MajidTAC11} improves the work in \cite{PolaAutom2008} by removing the stability assumption; hence, it can be applied to the models considered in this paper. 
However, while the results in \cite{MajidTAC11} are based on a uniform discretization of the state space which can imply a large computational load, our results avoid this problem by working directly with the initial partition of the PWA systems, and by refining these sets step--by--step.  
As an example, we computed an abstraction of the PWA system $\Sigma$ in Section 7 by adapting the results of \cite{MajidTAC11} and compared it with  $\mathbb{A}_{1}(\Sigma)$. In order to get a resolution that is comparable with the one of $\mathbb{A}_{1}(\Sigma)$, we select the precision $\varepsilon=\min_{i\in [1;14]}\varepsilon_{i}= 0.125$, where $\varepsilon_{i}$ is the minimal length of the sides of the smallest hyperrectangle containing the polytope $X_{1}^{i}\in \mathbb{X}_{1}$ ($i\in [1;14]$). With this choice of $\varepsilon$ we obtained an abstraction consisting of $154$ states (vs. $14$ states of $\mathbb{A}_{1}(\Sigma)$) and $3248$ transitions (vs. $30$ transitions of $\mathbb{A}_{1}(\Sigma)$).
In future work we plan to develop efficient computational tools to construct the proposed abstractions and controllers. Useful insights in this direction are reported in \cite{Belta2012,PolaTACTN2011}.

\bibliographystyle{alpha}
\bibliography{biblio1}

\end{document}